\documentclass[twoside,final]{IEEEtran}
\IEEEoverridecommandlockouts

\usepackage{enumitem,eurosym,float,graphicx,lettrine,mathrsfs,multicol,multirow,nomencl,pict2e,psfrag}
\usepackage{ragged2e,setspace,subfigure,supertabular,tabularx,url,epstopdf}
\usepackage[USenglish]{babel}
\usepackage{multirow}

\usepackage{algorithm,algorithmic,amssymb,amsmath,amssymb,amsthm,array,bibentry,cite,color,comment,enumerate}
\usepackage{flushend}
\usepackage{booktabs}
\usepackage{epstopdf}
\usepackage{booktabs}
\usepackage{balance}
\usepackage[outercaption]{sidecap}
{\newtheorem{proposition}{Proposition}}
\begin{document}
\title{Secrecy Outage Analysis over Correlated\\ Composite Nakagami-$m$/Gamma Fading Channels}
\author{George C. Alexandropoulos,~\IEEEmembership{Senior~Member,~IEEE} and Kostas~P.~Peppas,~\IEEEmembership{Senior~Member,~IEEE}
\thanks{G. C. Alexandropoulos is with the Mathematical and Algorithmic Sciences Lab, Paris Research Center, Huawei Technologies France SASU, 92100 Boulogne-Billancourt, France. The views expressed here are his own and do not represent Huawei's ones. (e-mail: george.alexandropoulos@huawei.com).}
\thanks{K. P. Peppas is with the Department of Telecommunication Science and Technology, University of Peloponnese, Tripoli 22100, Greece (e-mail: peppas@uop.gr).}
}

\maketitle

\begin{abstract}
The secrecy outage performance of wireless communication systems operating over spatially correlated composite fading channels is analyzed in this paper. We adopt a multiplicative composite channel model for both the legitimate communication link and the link between the eavesdropper and the legitimate transmitter, consisting of Nakagami-$m$ distributed small-scale fading and shadowing (large-scale fading) modeled by the Gamma distribution. We consider the realistic case where small-scale fading between the links is independent, but shadowing is arbitrarily correlated, and present novel analytical expressions for the probability that the secrecy capacity falls below a target secrecy rate. The presented numerically evaluated results, verified by equivalent computer simulations, offer useful insights on the impact of shadowing correlation and composite fading parameters on the system's secrecy outage performance.
\end{abstract}

\section{Introduction} \label{sec:Intro}
\noindent
Physical Layer Security (PLS) has been recently considered as a companion technology to conventional cryptography offering the potential to significantly enhance the quality of secure communication in fifth generation (5G) wireless networks \cite{Yang_ComMag_2015_all}. In the pioneering work of Wyner in information theoretic security \cite{Wyner_Bell_1975}, it was shown that secure communication is feasible when the channel quality of legitimate parties is better than that of the eavesdropper. However, in practice, there are certain cases where the latter channels may experience correlated conditions, which will intuitively render the performance of PLS schemes limited. Spatial fading correlation highly depends on antenna deployments, proximity of the legitimate receiver and eavesdropper, as well as scatters around them \cite{Shiu_TCOM_2000_all}.

Assuming that the legitimate transmitter knows the channel gains towards the legitimate receiver and eavesdropper in \cite{Jeon_TIT_2011_all}, the loss of the secrecy capacity due to spatial correlation was quantified. Infinite series expressions for both the average secrecy capacity and outage probability were obtained in \cite{Sun_spl_2012_all} for correlated Rayleigh fading channels. By considering that the legitimate communication link and the link between the eavesdropper and the legitimate transmitter are arbitrarily correlated and both modeled by the log-normal distribution, \cite{Liu_coml_2013} studied the Probability of the Non Zero Secrecy Capacity (PNZSC). The PLS of Multiple-Input Multiple-Output (MIMO) wiretap channels with orthogonal space-time block codes was investigated in \cite{Ferdinand_2013_all}. In that work, the fading channels between the legitimate link and the link between the eavesdropper and the legitimate transmitter were assumed to be independent and modeled as Ricean and Rayleigh distributed, respectively. However, within each communication link the multiple fading channels, resulting from the utilization of multiple antennas, were assumed to be arbitrarily correlated. Recently in \cite{Pan_TVT_2016_all}, the average secrecy capacity and Secrecy Outage Probability (SOP) were studied for the cases where legitimate and eavesdropper links experience independent log-normal fading, correlated log-normal fading, or independent composite fading conditions. In the context of underlying cognitive radio networks, the SOP performance was also lately investigated in \cite{J:ZhangJ_all} considering correlated Rayleigh fading.

Motivated by the latest advances in the secrecy capacity analysis \cite{J:Cuma_George_TVT_2017_all} and aiming at studying PLS performance under more realistic fading conditions, we adopt in this paper a correlated composite fading channel model for the legitimate and eavesdropping links. Our model comprises of independent small-scale fading and arbitrarily correlated shadowing. For the small-scale fading we consider the versatile Nakagami-$m$ fading model \cite{J:Wang_MIMO_Nak_2014_all}, while shadowing (large-scale fading) is modeled by the Gamma distribution. We first present a novel analytical expression for the numerical SOP evaluation. Then, for the important special case of non zero secrecy capacity, a novel infinite series representation for PNZSC is deduced. Finally, in order to obtain further insights on the key factors affecting PLS performance, a simple closed form expression for PNZSC that becomes asymptotically tight for high values of the Signal-Noise-Ratio (SNR) is presented. All derived analytical results are substantiated with equivalent ones obtained by means of computer simulations.

\emph{Notations}: $\mathbb{E}\{\cdot\}$ denotes expectation, $\Gamma(\cdot)$ is the Gamma function \cite[eq$.$ (8.310/1)]{B:Gra_Ryz_Book}, $(x)_i \triangleq \Gamma(x+i)/\Gamma(x)$ is the Pochhammer's symbol \cite[p$.$ xliii]{B:Gra_Ryz_Book}, $u(\cdot)$ is the unit step function \cite[p$.$ xliv]{B:Gra_Ryz_Book}, $K_a (\cdot)$ is the modified Bessel function of the second kind and order $a$ \cite[eq$.$ (8.407/1)]{B:Gra_Ryz_Book}, $U(\cdot,\cdot,\cdot)$ is the Kummer hypergeometric function \cite[eq$.$ (9.210/2)]{B:Gra_Ryz_Book}, and $G\substack{m,n\\p,q}[\cdot]$ is the Meijer's G-function \cite[eq$.$ (9.301)]{B:Gra_Ryz_Book}.

\section{System and Channel Models} \label{sec:Models}
We consider a legitimate wireless communication link where a legitimate transmitter sends a message to the legitimate receiver $B$, while the eavesdropper $E$ attempts to decode this message from its received signal through the wireless link between itself and the legitimate transmitter. The channel links are assumed to be arbitrarily correlated due to either close proximity of $B$ and $E$ or similarity of the scatters around them. In addition, we assume that both channels experience ergodic block fading, where channel coefficients remain constants during a block period and vary independently from one block to the next one. We also consider, similar to \cite{Jeon_TIT_2011_all, Sun_spl_2012_all, Liu_wcoml_2013, Liu_coml_2013, J:ZhangJ_all, J:Wang_MIMO_Nak_2014_all, Ferdinand_2013_all, Pan_TVT_2016_all, J:HongjiangLei, J:CLiu1, J:HongjiangLei2, J:CLiu2}, that the channel coefficients from the legitimate transmitter to $B$ and to $E$ are ideally estimated in $B$ and $E$, respectively. In cases of active eavesdropping, $E$ is capable of estimating its corresponding channel as $B$ does, whereas in other cases, it needs to eavesdrop characteristics of the channel estimation process (e$.$g$.$, the legitimate transmitter's pilots signals).

Assuming narrowband communication links, the baseband received complex-valued signals at $B$ and $E$, respectively, can be mathematically expressed as
\begin{subequations}\label{eq:Received_signals}
\begin{equation}\label{eq:Received_signal_B}
y_{B} = \sqrt{p}h_{B}s + n_{B},
\end{equation}
\begin{equation}\label{eq:Received_signal_E}
y_{E} = \sqrt{p}h_{E}s + n_{E},
\end{equation}
\end{subequations}
where $p$ denotes the fixed average power of the legitimate transmitter and $s$ is its unit power complex-valued information message chosen from a discrete modulation set. In \eqref{eq:Received_signals}, $h_{B}$ and $h_{E}$ represent the complex channel gains from the legitimate transmitter to $B$ and to $E$, respectively. Also, $n_{B}$ and $n_{E}$ denote the zero mean Additive White Gaussian Noises (AWGNs) at $B$ and $E$, respectively, with variances $\sigma_B^2$ and $\sigma_E^2$.

Both wireless channels are assumed to be subject to composite propagation conditions incorporating multipath fading and shadowing. The former is modeled by the versatile Nakagami-$m$ distribution, while the latter by the Gamma distribution. In mathematical representation, we model the amplitudes of the channel gains as $g_{B}\triangleq|h_{B}|=\sqrt{b_1}w_1$ and $g_{E}\triangleq |h_{E}|= \sqrt{b_2}w_2$, where $b_1$ and $b_2$ are Gamma random variables (RVs) with shaping parameters $k_1$ and $k_2$ and scaling parameters $\theta_1$ and $\theta_2$, respectively. In addition, $w_1$ and $w_2$ are assumed to be Nakagami-$m$ RVs with shaping parameters $m_1$ and $m_2$ and average powers $\Omega_1\triangleq\mathbb{E}\{w_1^2\}$ and $\Omega_2\triangleq\mathbb{E}\{w_2^2\}$, respectively. Due to either close proximity of $B$ and $E$ and/or similarity of the scatters around them, we consider the realistic case where $g _{B}$ and $g_{E}$ are correlated RVs resulting from correlated shadowing, but small-scale fading is assumed to be independent between $B$ and $E$. As such, $w_1$ and $w_2$ are assumed to be independent Nakagami-$m$ RVs, whereas $b_1$ and $b_2$ are modeled as correlated Gamma RVs.

Capitalizing on the system model of \eqref{eq:Received_signal_B}, the instantaneous received SNR at $B$ is given by $\gamma_{1}\triangleq pg_{B}^2/\sigma_B^2$ with average value derived as $\overline{\gamma}_{1}\triangleq p\mathbb{E}\{g_B^2\}/\sigma_B^2$, where $\mathbb{E}\{g_B^2\}=k_1\theta_1\Omega_1$. Similarly from \eqref{eq:Received_signal_E}, the instantaneous received SNR at $E$ and its average value are given by $\gamma_{2}\triangleq pg_{E}^2/\sigma_E^2$ and $\overline{\gamma}_{2}\triangleq pk_2\theta_2\Omega_2/\sigma_B^2$, respectively. The joint Probability Density Function (PDF) of $\gamma_{1}$ and $\gamma_{2}$ for the considered arbitarily correlated composite Nakagami-$m$/Gamma fading channel model can be obtained by employing \cite[eq$.$ (3)]{Bithas_BiK_2009} for the special case of independent Nakagami-$m$ RVs and after using a standard transformation of RVs, yielding
\begin{align} \label{eq:Bivariate_SNRs}
\nonumber f_{\gamma_{1},\gamma_{2}}(x_1,x_2)  =&  \frac{4(1-\rho)^{k_2}}{\Gamma(m_1) \Gamma(m_2)}\sum_{i,j = 0}^{\infty}\frac{(k_1)_i(k_2-k_1)_j\rho^{i+j}}{i!j!(i+k_2)_j}
\\& \times\prod_{\ell = 1}^2\frac{A_\ell^{\xi_\ell}x_\ell^{ \xi_\ell-1}}{\Gamma(i+k_\ell)}K_{\psi_\ell}\left(2\sqrt{A_\ell x_\ell}\right),
\end{align}
where $\xi_1\triangleq(m_1+k_1+i)/2$, $\xi_2\triangleq(m_2+k_2+i+j)/2$, $\psi_1\triangleq m_1-k_1-i$, $\psi_2\triangleq m_2-k_2-i-j$, and $A_\ell\triangleq\frac{m_\ell k_\ell}{(1-\rho)\overline{\gamma}_\ell}$. In the latter PDF expression, $\rho\in\left[0,1\right)$ represents the correlation coefficient between the RVs $b_1$ and $b_2$ \cite[Sec$.$ II]{Bithas_BiK_2009}.

\section{Secrecy Performance Analysis} \label{sec:OP}
In this section, we present novel analytical expressions for the SOP and PNZSC performance of the considered PLS communication system operating over arbitrarily correlated composite Nakagami-$m$/Gamma fading channels.
\subsection{Secrecy Outage Probability (SOP)}
The SOP performance of the PLS system described in Section~\ref{sec:Models} is given by the following probability \cite[eq$.$ (7)]{Liu_wcoml_2013}
\begin{align}\label{eq:SOC}
P_{\rm o}(r) &\triangleq 1 - {\rm Pr}\left[\gamma_1>2^r\left(1+\gamma_2\right)-1\right] \nonumber\\
& =1- \int_0^\infty\int_{h(x_2,r)}^\infty f_{\gamma_{1},\gamma_{2}}(x_1,x_2){\rm d}x_1{\rm d}x_2,
\end{align}
where $r$ denotes the target secrecy rate in bps/Hz and $h(x_2,r)\triangleq (1+x_2)2^r-1$. Based on the latter integral expression, we establish in the following proposition a method for the efficient numerical SOP evaluation.
\begin{proposition}\label{Prop:Cmnumerical}\rm{
The SOP of the considered PLS system can be tightly approximated numerically using the expression given by \eqref{Eq:Secrecyoutage} (top of next page),
\begin{figure*}[!t]
\begin{equation}\label{Eq:Secrecyoutage}
\begin{split}
P_{\rm o}(r)  \cong& 1-4\sqrt{\pi}\frac{(1-\rho)^{k_2}}{\Gamma(m_1) \Gamma(m_2)}\sum_{i,j = 0}^{\infty}\frac{(k_1)_i(k_2-k_1)_j2^{-2i-2k_1+1}\rho^{i+j}}{i!j!(i+k_2)_j\Gamma(i+k_1)\Gamma(i+k_2)}\sum_{k=1}^{15}w_k t_k^{4m_1-1} \\
& \times G\substack{3,0\\1,3} \left[
\frac{2^{r-2}k_2m_2 t_k^4\overline{\gamma}_1}{k_1m_1\overline{\gamma}_2}+\frac{(2^r-1)k_2m_2}{(1-\rho)\overline{\gamma}_2}
  \Big| \substack{ 1 \\ 0, i+j+k_2, m_2}    \right]U\left(m_1-i-k_1+1/2,2m_1-2i-2k_1+1,2t_k^2\right)
\end{split}
\end{equation}
\hrulefill 
\end{figure*}
where $w_k$ and $t_k$ for $k=1,2,\ldots,15$ are the weights and abscissas given in \cite[Tabs$.$ II and III]{J:Steen_all}.}
\end{proposition}
\begin{proof}
Substituting the joint PDF of $\gamma_{1}$ and $\gamma_{2}$ given by \eqref{eq:Bivariate_SNRs} into \eqref{eq:SOC}, the following two-fold integral is deduced
\begin{equation}\label{Eq:I1}
\mathcal{I} \!=\! \int\limits_{0}^\infty\!\int\limits_{0}^\infty u\!\left(\frac{x_1}{h(x_2,r)}-1\right)
\!\prod_{\ell=1}^2 x_\ell^{ \xi_\ell-1}\!K_{\psi_\ell}\!\!\left(2\sqrt{A_\ell x_\ell}\right)\!{\rm d}x_\ell.
\end{equation}
\relpenalty=10000
\binoppenalty=10000
The inner integral, i$.$e$.$, the one with respect to $x_1$, can be computed in closed form by expressing the Bessel and unit step functions in terms of Meijer's G-functions, i$.$e$.$, as $K_\nu(2\sqrt{x}) = 0.5\sqrt{\pi}G\substack{2,0\\0,2} \left[ x  \left| \substack{ - \\ -\nu/2, \nu/2} \right.\right]$ \cite[eq$.$ (8.4.23/1)]{B:Prudnikov3_all} and $u\left(x-1\right) = G\substack{0,1\\1,1} \left[ x \left| \substack{ 1 \\ 0} \right.\right]$ \cite[eq. (8.4.2/1)]{B:Prudnikov3_all}, respectively. Then, by employing the integral expression \cite[eq$.$ (2.24.1/1)]{B:Prudnikov3_all}, \eqref{Eq:I1} can be simplified to the following single integral
\begin{equation}\label{Eq:I2}
\begin{split}
\mathcal{I} =& \frac{A_1^{-\xi_1}}{2}\int_0^\infty G\substack{3,0\\1,3} \left[ A_1 h(x_2,r)  \left| \substack{ 1 \\ 0, -\psi_1/2+\xi_1, \psi_1/2+\xi_1} \right.\right]\\
& \times x_2^{ \xi_2-1}K_{\psi_2}\left(2\sqrt{A_2 x_2}\right)
{\rm d}x_2.
\end{split}
\end{equation}
\begin{figure*}[!t]
\begin{equation}\label{Eq:Nzoutage}
\begin{split}
P_{\rm o}(0) =  \frac{(1-\rho)^{k_2}}{\Gamma(m_1) \Gamma(m_2)}\sum_{i,j = 0}^{\infty}\frac{(k_1)_i(k_2-k_1)_j\rho^{i+j}}{i!j!(i+k_2)_j\Gamma(i+k_1)\Gamma(i+k_2)}  
G\substack{2,3\\3,3} \left[ \frac{m_1 k_1 \overline{\gamma}_2}{m_2 k_2 \overline{\gamma}_1}  \left| \substack{ 1, 1-m_2, 1-i-j-k_2 \\ i+k_1,  m_1, 0}     \right.\right]
\end{split}
\end{equation}
\hrulefill 
\end{figure*}
It is noted that all necessary conditions for the existence of \cite[eq$.$ (2.24.1/1)]{B:Prudnikov3_all} are satisfied throughout this paper's analysis. The integral in \eqref{Eq:I2} cannot be in general solved in closed form when $r>0$ holds. However, by employing the identity $K_\nu(x) = \sqrt{\pi}e^{-x}(2x)^{\nu}U\left(0.5+\nu,1+2\nu, 2x \right)$ \cite[eq$.$ (9.238/3)]{B:Gra_Ryz_Book} as well as the change of variables $2\sqrt{A_2x_2} = y^2$, the resulting integral can be efficiently evaluated numerically by using the modified Gauss-Chebyshev quadrature technique described in \cite{J:Steen_all}. Following this technique, SOP can be numerically evaluated as in \eqref{Eq:Secrecyoutage}, thus, completing the proof.
\end{proof}

\subsection{Probability of Non Zero Secrecy Capacity (PNZSC)}
PNZSC defined using \eqref{eq:SOC} as $P_{\rm o}(0)$ often serves as a fundamental benchmark on the secrecy performance of PLS systems \cite{Liu_wcoml_2013}. Although it can be numerically approximated for the considered PLS system from \eqref{Eq:Secrecyoutage} after setting $r=0$, we next present a novel analytical PNZSC infinite series representation.
\begin{proposition}\label{Prop:PNZSC}\rm{
An infinite series expression for PNZSC for the considered PLS system is given by \eqref{Eq:Nzoutage} (top of this page).}
\end{proposition}
\begin{proof}
Starting from \eqref{eq:SOC}, the PNZSC $P_{\rm o}(0)$ is obtained as
\begin{equation}\label{Eq:Nzoutage2}
P_{\rm o}(r) = \int_0^\infty\int_{x_1}^\infty f_{\gamma_{1},\gamma_{2}}(x_1,x_2){\rm d}x_2{\rm d}x_1.
\end{equation}
Substituting the joint PDF expression \eqref{eq:Bivariate_SNRs} into \eqref{Eq:Nzoutage2}, the following two-fold integral appears in the PNZSC expression
\begin{equation}\label{Eq:J1}
\begin{split}
\mathcal{J} &= \int_0^\infty\int_{x_1}^\infty \prod_{\ell = 1}^2x_\ell^{ \xi_\ell-1}K_{\psi_\ell}\left(2\sqrt{A_\ell x_\ell}\right)
{\rm d}x_\ell.
\end{split}
\end{equation}
The latter inner integral with respect to $x_2$ can be solved using \cite[eqs$.$ (8.4.23/1), (8.4.2/1), and (2.24.1/1)]{B:Prudnikov3_all} yielding
\begin{equation}\label{Eq:J2}
\begin{split}
\mathcal{J} = &
\frac{A_2^{-\xi_2}}{2}\int_0^\infty G\substack{3,0\\1,3} \left[ A_2 x_1  \left| \substack{ 1 \\ 0, -\psi_2/2+\xi_2, \psi_2/2+\xi_2} \right.\right]\\
& \times x_1^{ \xi_1-1}K_{\psi_1}\left(2\sqrt{A_1 x_1}\right){\rm d}x_1.
\end{split}
\end{equation}
Finally, \eqref{Eq:Nzoutage} is deduced after using \cite[eq$.$ (2.24.1/1)]{B:Prudnikov3_all}.
\end{proof}

\subsection{Asymptotic Analysis for PNZSC}
To gain further insights on the impact of the composite fading parameters as well as of shadowing correlation on the considered PLS system's performance, we next present a closed form asymptotic expression for PNZSC that is valid for high values of the average received SNRs.
\begin{proposition}\label{Prop:asym}\rm{
For high values of $\overline{\gamma}_1$, PNZSC can be obtained from the following expression with $\alpha_1 \triangleq \min\{k_1, m_1\}$:

\begin{align}\label{eq:PNZSC_approx}
\nonumber P_{\rm o}(0) \overset{\overline{\gamma}_1 \rightarrow \infty}{\cong}&
\frac{(1-\rho)^{k_2}\Gamma(|k_1-m_1|)\Gamma(k_2+\alpha_1)}{\alpha_1\prod_{\ell=1}^2\Gamma(m_\ell)\Gamma(k_\ell)}
\\& \times \Gamma(m_2+\alpha_1)\left(\frac{m_1k_1\overline{\gamma}_2}{m_2k_2\overline{\gamma}_1}\right)^{\alpha_1}.
\end{align}}
\end{proposition}
\begin{proof}
When $\overline{\gamma}_1\rightarrow\infty$, holds $\theta_1\rightarrow\infty$. For this asymptotic case, the joint Moment Generating Function (MGF) of $b_1$ and $b_2$ obtained using \cite[eq$.$ (7)]{J:BivariateEL_all} can be approximated as
\begin{align}\label{eq:MGFjoint2}
\mathcal{M}_{b_1, b_2}(s_1, s_2) \overset{\overline{\gamma}_1 \rightarrow \infty}{\cong}& \frac{(1-\rho)^{-k_1}}{\prod_{\ell=1}^2\theta_\ell^{k_\ell}\left(s_\ell+\frac{1}{\theta_\ell(1-\rho)}\right)^{k_\ell}}.
\end{align}
By taking the inverse Laplace transform of the latter MGF, the joint PDF of $b_1$ and $b_2$ can be asymptotically approximated as
\begin{align}\label{eq:Bivariate_Gamma_apprpx}
f_{b_1,b_2}(y_1,y_2) \overset{\overline{\gamma}_1 \rightarrow \infty}{\cong}&  \frac{(1-\rho)^{-k_1}y_1^{k_1-1}y_2^{k_2-1}}{\Gamma(k_1)\Gamma(k_2)\theta_1^{k_1}\theta_2^{k_2}}
e^{-\sum_{\ell=1}^2\frac{y_\ell}{\theta_\ell(1-\rho)}}.
\end{align}
The joint PDF of RVs $g _{B}$ and $g_{E}$ can be derived as follows
\begin{align}\label{eq:Bivariate_Gains}
\nonumber f_{g_{B},g_{E}}(x_1,x_2) = \int_0^\infty\int_0^\infty &   f_{g_{B}|b_1}(x_1|y_1)f_{g_{E}|b_2}(x_2|y_2)
\\&\times f_{b_1,b_2}(y_1,y_2) {\rm d}y_1{\rm d}y_2,
\end{align}
where $f_{g_{B}|b_1}(\cdot)$ denotes the PDF of $g_{B}$ conditioned on $b_1$ and $f_{g_{E}|b_2}(\cdot)$ denotes the PDF of $g_{E}$ conditioned on $b_2$. Based on the channel model in Section~\ref{sec:Models}, the latter PDFs are the marginal Nakagami-$m$ PDFs with average powers $\mathbb{E}\{g_B^2|b_1\}=b_1\Omega_1$ and $\mathbb{E}\{g_E^2|b_2\}=b_2\Omega_2$, respectively. Using the transformations of RVs $\gamma_1 = g_B^2\overline{\gamma}_1/k_1\theta_1\Omega_1$ and $\gamma_2 = g_E^2\overline{\gamma}_2/k_2\theta_2\Omega_2$ in the joint PDF definition \eqref{eq:Bivariate_Gains} yields after some algebraic manipulations the following asymptotically approximate bivariate PDF expression
\begin{align}\label{eq:Bivariate_SNRs_2}
\nonumber& f_{\gamma_{1},\gamma_{2}}(x_1,x_2)  \overset{\overline{\gamma}_1 \rightarrow \infty}{\cong} 4(1-\rho)^{k_2} \\
& \times \prod_{\ell = 1}^2\left[ A_\ell^{\tilde{\xi}_\ell}
\frac{x_\ell^{ \tilde{\xi}_\ell-1}}{\Gamma(k_\ell)\Gamma(m_\ell)}
K_{\tilde{\psi}_\ell}\left(2\sqrt{A_\ell x_\ell}\right)\right],
\end{align}
where $\tilde{\xi}_\ell\triangleq(m_\ell+k_\ell)/2$ and $\tilde{\psi}_\ell\triangleq m_\ell-k_\ell$. The proof completes by using the identity $K_{\alpha_1}(x) \overset{x \rightarrow 0}{\cong} (2/x)^{|\alpha_1|}\Gamma(|\alpha_1|)/2$, a similar line of arguments as in the proof of Proposition~\ref{Prop:PNZSC}, and \cite[eq$.$ (2.24.2/1)]{B:Prudnikov3_all} for evaluating the finally resulting single integral with respect to $x_1$.
\end{proof}
\enlargethispage{\baselineskip}

\vspace{-3mm}
\section{Numerical Results and Discussion} \label{sec:Results}
In this section, we numerically evaluate the analytical expressions \eqref{Eq:Secrecyoutage}, \eqref{Eq:Nzoutage}, and \eqref{eq:PNZSC_approx} for the secrecy outage performance of the considered PLS communication system that operates over arbitrarily correlated composite Nakagami-$m$/Gamma fading channels. In the two figures that follow we also include equivalent results obtained by means of computer simulations in order to verify the correctness of the presented mathematical formulas. For the numerical evaluation of the double infinite series appearing in \eqref{Eq:Secrecyoutage} and \eqref{Eq:Nzoutage}, we have truncated both series in each expression to the same finite number of terms $\mathcal{N}$ leading to a perfect match with equivalent computer simulations up to the third significant digit. In general, $\mathcal{N}$ increases with increasing values of any of the parameters $\rho$, $m_1$, $m_2$, $k_1$, and $k_2$, and decreases as the average SNR increases. To further decrease the computational complexity of \eqref{Eq:Secrecyoutage} and \eqref{Eq:Nzoutage}, the included Kummer hypergeometric function and the Meijer G-function have been first precomputed and then stored, and finally used in the evaluation of the respective truncated series.

\begin{figure}[!t]
\centering
\includegraphics[keepaspectratio,width=2.8in]{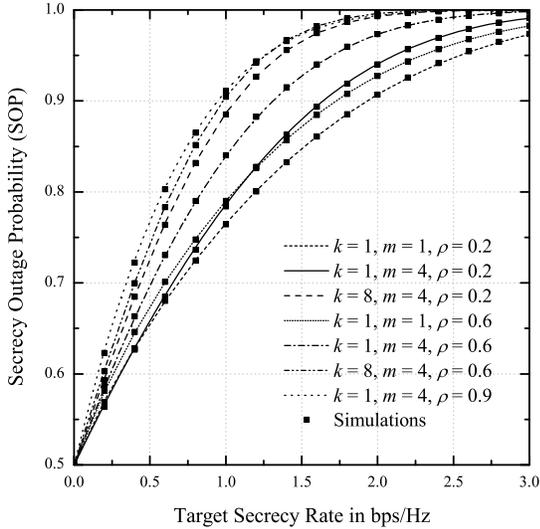}
\caption{SOP vs $r$ in bps/Hz for various values of the correlation coefficient $\rho$, the small-scale shaping parameter $m$, and the shadowing parameter $k$.}\label{Fig:poutvrate}\vspace{-3mm}
\end{figure}
Figure~\ref{Fig:poutvrate} illustrates SOP versus $r$ in bps/Hz for the common average SNR values $\overline{\gamma}_1=\overline{\gamma}_2=4$dB, various values of the correlation coefficient $\rho$, of the common shadowing shaping parameter $k\triangleq k_1=k_2$, and of the common small-scale shaping parameter $m\triangleq m_1=m_2$. PNZSC as a function of $\overline{\gamma}_1$ in dB is depicted in Fig$.$~\ref{Fig:poutvssnr} for $\overline{\gamma}_2= 0$dB, $m=4$, as well as different values of $\rho$ and the shadowing parameters $k_1$ and $k_2$. For the SOP results in Fig$.$~\ref{Fig:poutvrate} we have used from $\mathcal{N}=10$ (for $k=1$, $m=1$, and $\rho=0.2$) to $\mathcal{N}=45$ (for $k=1$, $m=4$, and $\rho=0.9$) terms to truncate both infinite series included in \eqref{Eq:Secrecyoutage}. The corresponding range of terms in Fig$.$~\ref{Fig:poutvssnr} for the PNZSC curves obtained using \eqref{Eq:Nzoutage} is from $\mathcal{N}=10$ to $\mathcal{N}=30$. As shown in both figures and as expected, SOP degrades with increasing $r$ and PNZSC improves with increasing $\overline{\gamma}_1$. In addition, increasing $\rho$ and/or the shadowing parameters degrades SOP for the plotted range of $r$ in Fig$.$~\ref{Fig:poutvrate}, and improves PNZSC as $\overline{\gamma}_1$ increases as depicted in Fig$.$~\ref{Fig:poutvssnr}. This trend for the SOP and PNZSC performance agrees with that in \cite{Jeon_TIT_2011_all, Sun_spl_2012_all, Liu_wcoml_2013, Liu_coml_2013, Ferdinand_2013_all, Pan_TVT_2016_all} either correlated small-scale fading or correlated shadowing was considered.

The numerically evaluated performance results of the analytical expressions \eqref{Eq:Secrecyoutage}, \eqref{Eq:Nzoutage}, and \eqref{eq:PNZSC_approx} included in Figs$.$~\ref{Fig:poutvrate} and~\ref{Fig:poutvssnr} reveal that large and severely correlated shadowing in the legitimate receiver and eavesdropper might have a detrimental effect in the secrecy outage performance, even if small-scale fading between these nodes is independent. Future extensions of our framework include the consideration of MIMO techniques at some or all communications ends and the analysis of the impact of imperfect channel estimation.
\enlargethispage{\baselineskip}
\begin{figure}[!t]
\centering
\includegraphics[keepaspectratio,width=2.8in]{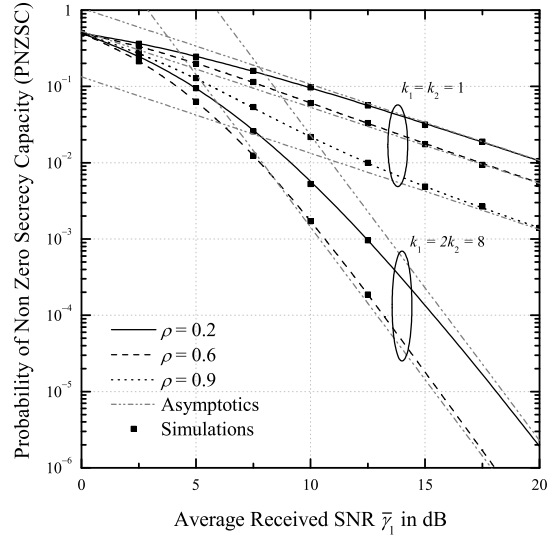}
\caption{PNZSC vs $\overline{\gamma}_1$ in dB for $\overline{\gamma}_2=0$dB and $m=4$, different values of the correlation coefficient $\rho$, and the shadowing parameters $k_1$ and $k_2$.}\label{Fig:poutvssnr}
\end{figure}


\vspace{-1.6mm}
\addcontentsline{toc}{chapter}{References}
\bibliographystyle{IEEEtran}
\bibliography{IEEEabrv,PLS_refs}
\enlargethispage{\baselineskip}

\end{document}